\documentclass[letterpaper, 10 pt, conference]{ieeeconf}  

\IEEEoverridecommandlockouts                              
\usepackage{amsmath,amssymb,amsfonts,bm, mathtools, dsfont}
\usepackage{xcolor}
\usepackage{cite}
\usepackage{tikz}
\usepackage{pgfplots}
\usepackage{pgfplotstable}
\usepackage{tabularray}
\usepackage{booktabs}
\usepackage{graphicx}
\usepackage{array}
\usepackage{pifont}   
\usepackage{threeparttable}
\usepackage{rotating}
\usepackage{multirow}
\usepackage{makecell}

\makeatletter
\let\proof\@undefined
\let\endproof\@undefined
\makeatother
\usepackage{amsthm}

\makeatletter
\let\MYcaption\@makecaption
\makeatother

\usepackage[font=footnotesize]{subcaption}

\makeatletter
\let\@makecaption\MYcaption
\makeatother
\newcommand{\cmark}{\ding{51}} 
\newcommand{\xmark}{\ding{55}} 

\renewenvironment{proof}[1][\proofname]{\par\itshape\noindent\textit{#1.}\enspace}{\par}

\usetikzlibrary{positioning}

\newtheorem{definition}{Definition}
\newtheorem{theorem}{Theorem}

\newtheorem{proposition}{Proposition}

\makeatletter
\def\@opargbegintheorem#1#2#3{\trivlist
  \item[]{#1\ #2\ (\itshape #3\bfseries): }\itshape}
\makeatother

\title{\LARGE \bf

Probably Approximately Correct Guarantees for Data-Driven Reachability Analysis: A Theoretical and Empirical Comparison

}

\author{Elizabeth Dietrich$^*$, Hanna Krasowski$^*$, and Murat Arcak 
\thanks{$^{*}$Equal contribution.}
\thanks{E. Dietrich, H. Krasowski, and M. Arcak are with the University of California, Berkeley, USA. Email: {\tt\small \{eadietri, krasowski, arcak\}@berkeley.edu}.}
\thanks{This work was funded in part by the National Science Foundation under grant CNS-2111688. Elizabeth Dietrich was also supported by an NSF Graduate Research Fellowship. The authors would like to thank Stephen Tu for helpful discussions that led to revisions of Sec. III.}
}

\begin{document}

\maketitle
\thispagestyle{empty}
\pagestyle{empty}

\begin{abstract}
Reachability analysis evaluates system safety by identifying the set in which states may evolve over a time horizon. 
Data-driven reachability analysis estimates reachable sets and derives probabilistic guarantees directly from data. 
Several popular techniques for validating reachable sets---conformal prediction, scenario optimization, and the holdout method---admit similar Probably Approximately Correct (PAC) guarantees. 
We establish a formal connection between these PAC bounds and present an empirical case study on reachable sets to illustrate the trade-offs associated with these methods. We argue that despite the formal relationship between these techniques, subtle differences arise in both the interpretation of guarantees and the parameterization. We conclude with practical advice on the usage of these methods.

\end{abstract}

\section{Introduction}
Data-driven reachability analysis enables safety evaluation of dynamical systems when explicit models are unavailable or too complex for model-based approaches. In contrast to model-based reachability, which ensures absolute coverage of all trajectories, data-driven methods provide probabilistic guarantees on the portion of the true reachable set captured.
Existing approaches broadly fall into two categories: (i)~those leveraging system dynamics for tighter bounds \cite{Griffioen2023, alanwar23-TAC, Tebjou2023}, and (ii) those imparting black-box system assumptions for computationally efficient, but looser bounds \cite{lindemann2025formalverificationcontrolconformal, pmlr-v242-dietrich24a, l4dc-comparison}.

Reachability approaches that exploit system structure include Gaussian processes \cite{Griffioen2023, Akametalu2014}, conformal prediction \cite{Hashemi2023, Tebjou2023}, and algorithms using zonotopes \cite{alanwar23-TAC} and ellipsoids \cite{Park2024}.
However, these approximations require data that scales exponentially with state dimension. 
Alternatively, data-driven approximations can be learned solely from trajectory data, resulting in a random variable whose support is the reachable set. The strongest guarantee achievable in this setting is a Probably Approximately Correct (PAC) bound \cite{pacbound, dietrich2025datadrivenreachabilityscenariooptimization}.
 
PAC bounds yield two layers of probability---a guarantee on the accuracy of the prediction and a guarantee on the confidence in sample generation. They are ubiquitous in statistical learning theory, as seen through several well-established methods including conformal prediction \cite{lindemann2025formalverificationcontrolconformal, angelopoulos2022gentleintroductionconformalprediction, Shafer2008}, scenario optimization \cite{dembo1991scenario, nonconvex24, scenario-sample-discard}, and the holdout method \cite{tempobook, langford2005tutorial}. 
These probabilistic approaches have attracted attention in control theory due to their ability to provide comparable finite-sample guarantees. Nevertheless, current formulations are often developed under distinct problem settings, hindering systematic comparison across methods. 

Recent research has examined connections between conformal prediction, scenario optimization, and the holdout method \cite{dietrich2025datadrivenreachabilityscenariooptimization, l4dc-comparison, osullivan2025bridgingconformalpredictionscenario}. However, the relationship between their PAC guarantees remains largely unexplored. The similarity between the training-conditional guarantees of conformal prediction and those of the holdout method \cite{langford2005tutorial} was first noted in \cite{pmlr-v25-vovk12}.
Subsequently, the relationship between conformal prediction and scenario optimization was investigated empirically in the context of reachability \cite{l4dc-comparison} and temporal logic specifications \cite{Coppola2024}, as well as theoretically for the expected probabilities of constraint violation~\cite{osullivan2025bridgingconformalpredictionscenario}.

In this paper, we explore the relationship between conformal prediction, holdout method, and scenario optimization with sample discarding for data-driven reachability analysis. Extending from \cite{osullivan2025bridgingconformalpredictionscenario}, we establish new connections between their PAC guarantees and provide insight into their roles in data-driven verification.
Our main contributions include:
\begin{itemize}
    \item We quantify the difference between the holdout method's and conformal prediction's reachable set accuracy as the difference between two beta inverse CDFs.
    \item We show that scenario optimization with sample discarding and conformal prediction admit parallel PAC guarantees and yield identical reachable set estimates.
    \item We provide an empirical case study on reachable sets and discuss practical implications for method selection.
\end{itemize}

In Sec.~\ref{sec:preliminaries}, we introduce the methods utilized in our derivations. In Sec.~\ref{sec:CPvsHO} and Sec.~\ref{sec:CPvsSO}, we derive formal connections between PAC bounds. Finally, in Sec.~\ref{sec:empiricaleval} we evaluate these equivalences through a numerical example.


\section{Preliminaries} \label{sec:preliminaries}

\begin{table*}[t]
\centering
\caption{Features of Probabilistic Methods
\vspace*{-1mm}
}
\label{tab:features}
\resizebox{.95\textwidth}{!}{
\begin{threeparttable}
\begin{tabular}{lccc}
\toprule
Feature &
\textbf{\shortstack{Scenario Optimization\\with Sample Discarding}}&
\textbf{\shortstack{Split Conformal Prediction\\with Training-Conditional Coverage}} &
\textbf{The Holdout Method}\\
\midrule
Convexity & \cmark & \xmark & \xmark \\
i.i.d. Data & \cmark & \cmark & \cmark \\
User-defined Parameters$^\dagger$ & $N, \epsilon, \beta$ & $K, \alpha, \beta$ & $M, \beta$ \\
Calculated Values$^\dagger$ & $k$ & $\hat{q}, C$ & $\epsilon$ \\
Analysis Type & \textit{ex post} & \textit{ex ante} & \textit{ex post} \\
Learns Reachable Set Predictor & \cmark & \xmark & \xmark \\
\bottomrule
\end{tabular}

\begin{tablenotes}
\footnotesize
\item[$\dagger$] 
$N=$ number of scenarios,
$\epsilon=$ accuracy,
$\beta=$ confidence,
$K=$ calibration set size,
$\alpha=$ user-defined error rate,
$M=$ test set size,
$k =$ number of scenarios to remove,
$\hat{q}=$ quantile of calibration scores,
$C=$ prediction set.
\vspace*{-5mm}
\end{tablenotes}
\end{threeparttable}
}
\end{table*}

We denote probabilities by $P$, and product probability over $\square$ samples by $P^\square$. Parameters from conformal prediction and scenario optimization are denoted by $CP$ and $SO$ subscripts.

\subsection{Data-Driven Reachability Analysis}
A forward reachable set is defined as $\mathcal{R}_t = \{\Phi(t;t_0, x_0, d) : x_0 \in \mathcal{X}_0, d \in \mathcal{D}\}$ where $\mathcal{X}_0 \subseteq \mathbb{R}^{n_x}$ is the set of initial states, $\mathcal{D}$ is the set of disturbance signals $d: [t_0, t] \rightarrow \mathbb{R}^{n_d}$, and $\Phi(t;t_0, \cdot, \cdot) : \mathcal{X}_0 \times \mathcal{D} \rightarrow \mathbb{R}^{n_x}$ is the state transition function.  
We aim to compute an approximation $\hat{\mathcal{R}}$ that is close to the true reachable set in a probabilistic sense. 

To describe such an approximation, we endow $\mathcal{X}_0$ and $\mathcal{D}$ with probability distributions $\mu_{\mathcal{X}_0}$ and $\mu_\mathcal{D}$, respectively and draw samples $\delta^{(i)} =\Phi(t_1;t_0, x_{0i}, d_i), i = 1, \dotsc, N $ where $x_{01}, \dotsc, x_{0N} \overset{i.i.d.}{\sim} \mu_{\mathcal{X}_0}$, $d_1, \dotsc, d_N \overset{i.i.d.}{\sim} \mu_{\mathcal{D}}$. We compute reachable set estimates in the form of 
\begin{equation}
\label{eq:reachability}
\begin{aligned}
\hat{\mathcal{R}}(\theta) = \{\delta \in \mathbb{R}^{n_x} : g(\delta, \theta) \leq 0\},
\end{aligned}
\end{equation} 
where $g : \mathbb{R}^{n_x} \times \mathbb{R}^{n_{\theta}} \rightarrow \mathbb{R}$, and $\theta$ represents a parameterization of the class of admissible reachable set estimators.

Given samples $\delta^{(1)}, \dots, \delta^{(N)}$, accuracy $\varepsilon$, and confidence $\beta$, we aim to find the minimum-volume reachable set that contains the samples and satisfies the probabilistic guarantee 
\begin{equation}
\label{eq:genpacbound}
    P^N ( V(\hat{\mathcal{R}}(\theta)) >\varepsilon) \leq \beta.
\end{equation} 
The violation probability $V(\hat{\mathcal{R}}(\theta))$ is the probability that an unseen scenario will violate the reachable set:
     \begin{equation}
       V(\hat{\mathcal{R}}(\theta)) \equiv  P\{ g(\delta, \theta) > 0\} \equiv P\{ \delta \not\in \hat{\mathcal{R}}(\theta) \}. 
        \label{eq:violationprob}
    \end{equation}

\subsection{Scenario Optimization}
Scenario optimization relaxes chance-constrained optimization problems by solving a non-probabilistic problem on a finite number of samples \cite{dembo1991scenario}. 
To calculate reachable set estimates of the form \eqref{eq:reachability}, we fix a functional ${\rm Vol}: \mathbb{R}^{n_{\theta}} \rightarrow \mathbb{R}_{\geq 0}$ as a proxy for the volume of $\hat{\mathcal{R}}(\theta)$, and solve:
\begin{equation}
\label{eq:volprox}
\begin{aligned}
& \underset{\theta \in \mathbb{R}^{n_{\theta}}}{\min}
& & \rm{Vol}(\theta) \\
& \text{ s.t.}
& & g(\delta^{(i)}, \theta) \leq 0, \quad \forall i = 1, \dotsc, N.
\end{aligned}
\end{equation} 
Although several solution methods exist \cite{pmlr-v120-devonport20a, pmlr-v242-dietrich24a},
we summarize the sampling-and-discarding approach \cite{scenario-sample-discard}, as this forms the basis of Sec.~\ref{sec:cptosowcd}.
This approach examines convex-constrained optimization problems and allows for violation of $k$ constraints to improve the solution. Given a scenario program~\eqref{eq:volprox}, the removal procedure is defined as follows:
\begin{definition}
[Scenario Discarding Algorithm]
\label{def:discardscenarios}
    An algorithm $\mathcal{A}$ for constraint removal 
    selects and removes $k < N$ constraints in which all $k$ constraints improve the solution of the scenario optimization program (i.e., active constraints). 
\end{definition}
\noindent 
This removal algorithm yields the following guarantee:
\begin{theorem} [Theorem 2.1, \cite{scenario-sample-discard}]
\label{thm:samplesdiscarded}
Let $\beta \in (0,1)$ be a given confidence parameter value. If $N$, $k$, and $\epsilon$ are such that 
\begin{equation}
\label{eq:sowd}
    \binom{k+n_\theta-1}{k}\sum^{k+n_\theta-1}_{i=0}\binom{N}{i}\epsilon^i(1-\epsilon)^{N-i} \leq \beta
\end{equation}
where $n_\theta$ is the number of optimization variables $\theta$, then Eq.~\eqref{eq:genpacbound} holds with $\varepsilon = \epsilon$.
\vspace{1mm}
\end{theorem}

Theorem \ref{thm:samplesdiscarded} establishes a relationship among $N, k, \epsilon, $ and $\beta$, permitting computation of the largest number $k$ of constraints for discardment. While there is no closed-form solution for $k$, Eq.~\eqref{eq:sowd} can be solved numerically. For an upper-bound on $k$, we utilize the following explicit formula \cite[Eq. (8)]{scenario-sample-discard}:
\begin{equation}
\label{eq:kremoval}
    k \leq \epsilon N - n_\theta + 1 - \sqrt{2\epsilon N \ln \frac{(\epsilon N)^{n_\theta-1}}{\beta}}.
\end{equation}

\subsection{Split Conformal Prediction}\label{subsec:split_conformal_prelim}
In this paper, we restrict attention to split conformal prediction, which uses a finite number of scenarios to construct coverage guarantees for an algorithm's predictions 
\cite{lindemann2025formalverificationcontrolconformal, pmlr-v25-vovk12, angelopoulos2022gentleintroductionconformalprediction}. 
Given a user-defined probability, conformal prediction performs a calibration step in which nonconformity scores are calculated for an additional dataset.
This results in an empirical quantile and construction of a prediction set, which contains an unseen scenario with the prescribed probability. This translates into the following procedure for a reachable set estimate $\hat{\mathcal{R}}(\theta)$ for a fixed $\theta$:
\begin{enumerate}
    \item Consider $\hat{\mathcal{R}}(\theta)$ and i.i.d. $\delta^{(1)}, \dots, \delta^{(K)}$. 
    \item Choose the score function $g(\delta, \theta)$, the user-defined error rate $\alpha$, and the size of the calibration set $K$.
    \item Compute $\hat{q}$, the $\frac{\lceil  (K+1)(1-\alpha)\rceil}{K}$ quantile of the calibration scores $c_1 = g(\delta^{(1)}, \theta), \dots, c_K = g(\delta^{(K)}, \theta)$.
    \item Form the prediction set,
    which is equivalent to updating the reachable set to ensure marginal coverage,
    \begin{equation}
        \hat{\mathcal{R}}_{CP}(\theta) = \{ \delta \in \mathbb{R}^{n_x}: g(\delta, \theta) \leq \hat{q}\}.
    \end{equation}
\end{enumerate}
This conformal procedure results in prediction sets with the following coverage guarantee: 
\begin{theorem}[Theorem 1, \cite{angelopoulos2022gentleintroductionconformalprediction}]
\label{thm:maincoverage}
    If $\delta^{(1)}, ..., \delta^{(K)}$ and $\delta^{(K+1)}$ are i.i.d., and $\hat{\mathcal{R}}_{CP}(\theta)$ was obtained as outlined above, then 
    \begin{equation}\label{eq:marginalcoverage}
        \mathbb{E}_{\{\delta\}^K_{i=1}}[P(\delta^{(K+1)} \in \hat{\mathcal{R}}_{CP}(\theta)| \{\delta\}^K_{i=1})] \geq 1 - \alpha.
    \end{equation}
\end{theorem}

The coverage of a conformal predictor conditionally on a calibration set of size $K$ is a random quantity \cite{angelopoulos2022gentleintroductionconformalprediction}. 
If a conformal prediction algorithm is run more than once with a freshly sampled calibration set, the resulting coverage will vary---an algorithm will achieve $1-\alpha$ coverage on average~\cite{angelopoulos2022gentleintroductionconformalprediction,angelopoulos2025theoreticalfoundationsconformalprediction, pmlr-v25-vovk12}. 
Further, the fluctuations between runs follow a Beta distribution $Beta$ with analytic form 
\begin{align}
\label{eq:calibrationcoverage}
    P (\delta^{(K+1)} \in \hat{\mathcal{R}}_{CP}(\theta) | \{\delta\}^K_{i=1})\sim  Beta&(K+1-l, l),\\
    & l = \lfloor (K+1) \alpha \rfloor. \notag
\end{align}
While conformal prediction typically provides guarantees of the form \eqref{eq:marginalcoverage}, in Sec. \ref{sec:CPvsHO} and \ref{sec:CPvsSO} we instead consider training-conditional coverage, which uses \eqref{eq:calibrationcoverage} to obtain a bound on the conditional violation probability for a given confidence level $\beta \in (0,1)$ \cite[Eq. 4.1]{angelopoulos2025theoreticalfoundationsconformalprediction}. This is equivalent to bounding the violation probability of $\hat{\mathcal{R}}_{CP}(\theta)$,
\begin{equation}
\label{eq:conformalbound}
    P^K(V(\hat{\mathcal{R}}_{CP}(\theta) | \{\delta\}^K_{i=1}) \leq \alpha ) \geq 1 - \beta.
\end{equation}
The tail probability of the coverage conditionally on the calibration data is controlled by $\beta$ \cite{angelopoulos2022gentleintroductionconformalprediction}. Therefore, for any $\beta$, the calibration set cardinality can be computed using a beta inverse cumulative distribution function (CDF).
\begin{definition}[Beta Inverse CDF] 
\label{def:betacdf}
Given beta distribution $Beta(a,b)$, $\eta \in (0, 1]$, and beta function $\mathbf{B}$,
\begin{align}
    \overline{\text{Beta}}(a, b, \eta) &= \underset{e}{\max}\bigl\{ e : \text{Beta}_{cdf}(a, b, e) \leq \eta \bigr\}, \text{where} \notag\\
    \text{Beta}_{cdf}(a, b, e) &= \frac{1}{\mathbf{B}(a,b)}\int_{0}^{e}t^{a-1}(1-t)^{b-1}dt.\\ \notag
\end{align}
\end{definition}
\vspace*{-2mm}
\vspace*{-3mm}
\subsection{The Holdout Method}
The holdout method utilizes a test set of $M$ fresh samples (i.e., distinct from training set) to provide a risk estimate for a given model \cite{tempobook, hardtrecht2022patterns, dietrich2025datadrivenreachabilityscenariooptimization}. We draw a new set of samples $\delta^{(i)}_s =\Phi(t_1;t_0, x_{0i}^s, d_i^s), i = 1, \dotsc, M $ where $x_{01}^s, \dotsc, x_{0M}^s \overset{i.i.d.}{\sim} \mu_{X_0}$, $d_1^s, \dotsc, d_M^s \overset{i.i.d.}{\sim} \mu_{D}$, and test the accuracy of $\hat{\mathcal{R}}(\theta)$ on $\delta_s^{(1)}, \dotsc, \delta_s^{(M)}$. We aim to calculate the probability of observing at most $k$ reachable set violations out of $M$ samples. 
This is a binomial tail bound, which we invert to bound the true error of the reachable set:
\begin{definition}[Binomial Tail Inversion]
\label{def:binomial_tail_inversion}
For $k$ violations out of $M$ scenarios and $\eta \in (0, 1]$:
    \begin{equation}
        \overline{\text{Bin}}(k, M, \eta) = \max_{e}\Bigl\{ e : \text{Bin}_{cdf}\Bigl(k, M, e \Bigr) \geq \eta \Bigr\}, \,\,\text{where}
        \notag
    \end{equation}
    \begin{equation}
        \text{Bin}_{cdf}\Bigl(k, M, e \Bigr) =  \sum_{j=0}^{k} \binom Mj e^j (1-e)^{M-j}.
        \label{eq:binominversion}
    \end{equation}
    \label{def1}
\end{definition}
\vspace{-2mm}
\noindent The following probability bound holds for the reachable set estimate $\hat{\mathcal{R}}(\theta)$:
\begin{theorem} [Theorem 1, \cite{dietrich2025datadrivenreachabilityscenariooptimization}]
\label{thm:holdout}
Given any $\beta \in (0, 1)$, empirical count of violations $k$ for reachable set estimate $\hat{\mathcal{R}}(\theta)$, and a test set of $M$ i.i.d. scenarios,
\begin{equation}
\begin{aligned}
\label{equ:probviolate}
 {P^M}( V(\hat{\mathcal{R}}(\theta)) > \overline{\text{Bin}}(k, M, \beta) ) \leq \beta.
\end{aligned}
\end{equation} 
\end{theorem}

\subsection{Summary of Methods}
\label{sec:summethods}
Scenario optimization with sample discarding and conformal prediction are flexible to user-defined levels of accuracy, while the holdout method relies solely on a posteriori observations. There exist scenario optimization methods (e.g., nonconvex \cite{nonconvex24}, post-design verification \cite{campicdc25}, constraint relaxation \cite{garatticdc23}) that rely on a posteriori evaluation. 
However, given the a priori nature of setting conformal parameters, it is not meaningful to compare conformal prediction and these forms of scenario optimization.

To illustrate the differences of the considered methods, we summarize the assumptions, user-defined parameters, and features in Table \ref{tab:features}. Specifically, scenario optimization with sample discarding is the only method capable of simultaneously providing a PAC guarantee and a learned predictor. In contrast, conformal prediction with training-conditional coverage and the holdout method \textit{only} offer frameworks for uncertainty quantification.
Since this paper focuses on PAC guarantees, all three methods assume i.i.d. data. Additionally, we restrict our analysis to convexly parameterized sets.

\textbf{Problem Statement:} 
Consider a set estimate $\hat{\mathcal{R}}(\theta)$ as in~(\ref{eq:reachability}), where $g$ is convex in $\theta$, a desired confidence $\beta$, and a process to obtain a finite number of i.i.d. samples.
We study three data-driven verification approaches that provide guarantees of the form \eqref{eq:genpacbound}, where the error bound $\varepsilon$ depends on the verification approach. Our aim is to identify conditions under which these approaches yield identical or structurally equivalent PAC bounds in terms of $\varepsilon$ and $\beta$.


\section{Conformal Prediction vs. Holdout Method}\label{sec:CPvsHO}

In this section, we establish a relationship between conformal prediction and the holdout method, identifying the difference between the reachable set accuracies of the two methods. 
For conformal prediction to provide a valid coverage guarantee, the marginal error rate $\alpha$ must be fixed a priori; i.e., it cannot depend on a random variable. Consequently, the number of reachable set violations $k$ identified by the holdout method, given $K$ samples, cannot be used to define the marginal error rate, as incorrectly done in \cite{l4dc-comparison}.  

To compare the two methods and account for the randomness in $k$, we aim to derive a conformal prediction guarantee that holds simultaneously for all possible numbers of reachable set violations $k\in \{0, \dots, K\}$.
We achieve this by applying a union bound over the $K+1$ values.
This allocates a probability of $\frac{\beta}{K+1}$ to each fixed-$k$ statement, ensuring the overall probability of the union bound is at most $\beta$. As we will show, this penalty results in a conformal prediction guarantee that is strictly more conservative than the holdout method. 

\begin{theorem}
\label{thm:conformaltoholdout}
  Given reachable set estimate $\hat{\mathcal{R}}(\theta)$ and $K$ i.i.d. samples, let $k$ denote a random variable counting the number of reachable set violations.
  Then, conformal prediction with training-conditional coverage and the holdout method yield reachable set accuracies, $\alpha$ and $\epsilon$, respectively, such that:
    \begin{align}\label{eq:betaaccuracies}
      \alpha = \epsilon + \big( \overline{\text{Beta}}(k+1, K-k, 1-\frac{\beta}{K+1})  \\
      - \overline{\text{Beta}}(k+1, K-k, 1-\beta)\big).\notag
  \end{align}
    \end{theorem}
\begin{proof}
Consider a conformal prediction problem with score function $g(\delta, \theta)$. 
In the special case of i.i.d. data, the conformal prediction program can be run at a nominal level $\alpha'$, which is the unique value that satisfies \cite{angelopoulos2025theoreticalfoundationsconformalprediction}:
\begin{equation}\label{eq:cdf_cp}
    \text{Beta}_{cdf}((1-\alpha')(K+1), \alpha'(K+1),(1-\alpha)) = \beta.
\end{equation}
Take nominal error rate $\alpha' = \frac{k'+1}{K+1}$ with fixed integer $k'$ and the prediction set $\hat{\mathcal{R}}_{CP}(\theta) $.
For this choice of $\alpha'$, the conformal threshold is given by the $(K-k')$-th order statistic of the calibration scores. 
The corresponding training-conditional random variable has Beta law~\eqref{eq:calibrationcoverage}. Therefore, the training-conditional guarantee, obtained by inverting the Beta CDF, holds for each fixed $k' \in \{0, \dots, K\}$ with probability $\frac{\beta}{K+1}$.
By the union bound, the following holds simultaneously $\forall k' \in \{ 0, \dots, K\}$ with probability $\beta$,
    \begin{align} \notag
        P^K &\big(\bigcup^K_{k'=0} V(\hat{\mathcal{R}}_{CP}(\theta)) > \alpha(k')\big) \notag \\
        &\leq \sum^K_{k'=0} P^K(V(\hat{\mathcal{R}}_{CP}(\theta)) > \alpha(k')) \leq \beta. \notag 
    \end{align}
The probability of any individual event is bounded by the probability of the union of events; therefore, given the event $k'=k$, the bound becomes:
\begin{equation}
    P^K(V(\hat{\mathcal{R}}_{CP}(\theta)) > \overline{\text{Beta}}(k+1, K-k, 1-\frac{\beta}{K+1})) \leq \beta. \notag
\end{equation}
Now, consider the holdout method and Thm.~\ref{thm:holdout} with $M=K$ samples, $k$ violations, and error rate $e$. The partial sum of probabilities for the binomial distribution is an incomplete Beta integral \cite{Jowett1963, pmlr-v25-vovk12}, 
\begin{align}
    \text{Bin}_{cdf}(k, K, e) 
    & = 1 - \text{Beta}_{cdf}(k+1, K-k, e). 
\end{align}
To solve for a bound on $e$, we compute a binomial tail inversion, as defined in Eq. \eqref{eq:binominversion}:
\begin{align}
    \overline{\text{Bin}}(k, K, \beta) &= \max_{e}\Bigl\{ e : \text{Bin}_{cdf}(k, K, e ) \geq \beta \Bigr\}  \\
    &= \max_{e}\Bigl\{ e : 1- \text{Beta}_{cdf}(k+1, K-k, e ) \geq \beta \Bigr\} \notag \\
    &= \overline{\text{Beta}}(k+1, K-k, 1-\beta) = \epsilon. \notag\label{eq:relationship_beta_bin}
\end{align}
Since $1 - \frac{\beta}{(K+1)} > 1 - \beta$ and $\overline{\text{Beta}}(a,b,\eta)$ is nondecreasing in $\eta$, it follows that $\alpha > \epsilon$. Formally,
  \begin{align}
      \alpha = \epsilon + \big( \overline{\text{Beta}}(k+1, K-k, 1-\frac{\beta}{K+1}) \notag \\
      - \overline{\text{Beta}}(k+1, K-k, 1-\beta)\big). \notag \qed
  \end{align}
\end{proof}
\textbf{Remark 1.} 
Conformal prediction can utilize the number of reachable set violations by weakening accuracy by a factor of $\frac{\beta}{K+1}$ through the union bound. Additionally, it adapts the reachable set to achieve marginal coverage, while the holdout method bounds the true error, a distinction explored in \cite{coppola2025trainingconditionalconformalpredictionbinomial}.


\section{Conformal Prediction vs. Scenario Optimization}\label{sec:CPvsSO}
We now establish a relationship between the PAC guarantees of conformal prediction with training-conditional coverage and scenario optimization with sample discarding. 
We show that these methods
render structurally parallel PAC bounds (i.e., identical parameterization with $\varepsilon$ and $\beta$) for equivalent reachable sets. 
To begin, we define a scenario discarding algorithm for a reachable set (adapted from \cite{osullivan2025bridgingconformalpredictionscenario}):
\begin{definition}[Reachable Set Scenario Discarding]
\label{def:removefromset}
    Consider a reachable set $\hat{\mathcal{R}}(\theta) = \{ \delta \in \mathrm{R}^{n_x} : g(\delta, \theta) \leq 0\}$ for fixed $\theta$. Define $c_i=g(\delta^{(i)}, \theta)$ and the one-dimensional program  
    \begin{equation} \notag
        \underset{\overline{c}}{\min}\; \overline{c} \quad\text{s.t.} \quad c_i \leq \overline{c}, \quad \forall i = 1, \dots, N-k.
    \end{equation}
    The $k$ largest $c_i$ are removed, such that the $N-k$ smallest $c_i$ are returned, and the maximum $\overline{c^*}$ is the solution.
\end{definition}
\label{sec:cptosowcd}
\noindent We design conformal parameters, $K, \alpha, \beta$ to match the corresponding scenario parameters $N, \epsilon, \beta$.

\begin{proposition}
\label{prop:1}
    Consider the PAC guarantee in Eq. \eqref{eq:genpacbound} with user-defined confidence $\beta$. For $S$ i.i.d. samples, the conformal and scenario procedures yield an equivalent  $\varepsilon$:
    \begin{align}
    S &= \underset{S'}{\min}\; S' \quad \text{s.t.} \quad \overline{\text{Beta}}(l, S'+1-l, 1-\beta) \leq \varepsilon, \notag \\
        &\qquad\qquad\qquad\quad\ \hspace{1mm} l = \lfloor (S'+1) \varepsilon \rfloor; \notag \\
        \varepsilon &= \Bigg(1+S \ln\frac{1}{\beta} \Bigg) - \Bigg(\sqrt{\Big(1 + S \ln \frac{1}{\beta}\Big)^2 - 1} \Bigg). \notag
    \end{align}
\end{proposition}

\begin{proof}
    We aim to jointly select a finite number $S$ of i.i.d. samples and an accuracy level $\varepsilon$ such that the conformal ($K, \alpha, \beta$) and scenario ($N, \epsilon, \beta$) formulations yield identically parameterized PAC bounds.
    Let $S$ denote the number of i.i.d. samples used in both procedures. 
    
    To ensure conformal training-conditional coverage, $S$ is chosen such that the conformal prediction set achieves $1-\varepsilon$ coverage with probability $1-\beta$, as defined in \cite{angelopoulos2022gentleintroductionconformalprediction}:
    \begin{align}\label{eq:calibcard}
        S &= \underset{S'}{\min}\; S' \quad \text{s.t.} \quad \overline{\text{Beta}}(l, S'+1-l, 1-\beta) \leq \varepsilon, \notag \\
        &\qquad\qquad\qquad\quad\ \hspace{1mm} l = \lfloor (S'+1) \varepsilon \rfloor.  
    \end{align}

    Next, $\varepsilon$ is designed such that both procedures retain the same number of scenarios, $S - k$, where $k$ samples are set to be removed as in Definition \ref{def:removefromset}. To enforce this, the threshold of the conformal quantile $\hat{q}$ (see Sec.~\ref{subsec:split_conformal_prelim}) is set equal to the normalized number of retained scenarios:
    \begin{equation}
        \frac{\lceil(S+1)(1-\varepsilon)\rceil}{S} = \frac{S-k}{S}.
    \end{equation}
    The ceiling function can be bounded as 
    \begin{align}
    \label{eq:ceilingrelax}
    \frac{(S+1)(1-\varepsilon')}{S} 
    &\leq \frac{\lceil  (S+1)(1-\varepsilon)\rceil}{S} \notag \\
    &< \frac{(S+1)(1-\varepsilon') +1}{S}.
    \end{align}
    Together with Eq. \eqref{eq:kremoval}, we derive
    \begin{align} 
    \frac{(S+1)(1-\varepsilon')}{S} = \frac{S-k}{S}, \notag\\
    (S+1)(1-\varepsilon') = S - \varepsilon' S + \sqrt{2\varepsilon' S \ln \frac{1}{\beta}}, \notag\\
    1 - \varepsilon' = \sqrt{2\varepsilon' S \ln \frac{1}{\beta}}, \notag\\
    (1 - \varepsilon')^2 = 2\varepsilon' S \ln \frac{1}{\beta}, \notag \\
    (\varepsilon')^2 -2\varepsilon' \Big(1 + S \ln (\frac{1}{\beta})\Big) +1 = 0, \notag \\
    \varepsilon' = \Bigg(1+S \ln\frac{1}{\beta} \Bigg) - \Bigg(\sqrt{\Big(1 + S \ln \frac{1}{\beta}\Big)^2 - 1} \Bigg). 
    \label{eq:alpha}
    \end{align}
    From the ceiling relaxation in Eq. \eqref{eq:ceilingrelax}, it follows that $\varepsilon' \leq \varepsilon < \varepsilon' + \frac{1}{S+1}$. Thus, we set $\varepsilon=\varepsilon'$.

    By \textit{simultaneously} solving Eq. \eqref{eq:calibcard} and Eq. \eqref{eq:alpha}, $S$ and $\varepsilon$ are obtained to ensure the conformal and scenario procedures retain an equal number of scenarios. Consequently, the conformal and scenario formulations share the same parameterization, $K=N=S$, $\alpha = \epsilon = \varepsilon$.
    \qed
\end{proof}

Given Proposition \ref{prop:1}, we utilize the identical parameterization for the conformal and scenario formulations to update the reachable set estimates.   
    
\begin{theorem}
\label{thm:thm5}
    For the parametrization of Proposition \ref{prop:1}, conformal prediction with training-conditional coverage and scenario optimization with sample discarding yield identical reachable set estimates (i.e., $\hat{\mathcal{R}}_{CP}(\theta) = \hat{\mathcal{R}}_{SO}(\theta)$).
    \end{theorem}
\begin{proof}
    Consider a convexly parameterized reachable set estimate $\hat{\mathcal{R}}(\theta)$ with fixed $\theta$. Given Eq. \eqref{eq:calibcard} and Eq. \eqref{eq:alpha}, the conformal procedure results in the prediction set:
    \begin{equation}
        \hat{\mathcal{R}}^*_{CP}(\theta) = \{ \delta \in \mathbb{R}^{n_x}: g(\delta, \theta) \leq 
        \hat{q}
        \},\notag
    \end{equation}
    where $\hat{q}$ was designed to have $S-k$ scenarios by setting the threshold $\hat{q}$ in Proposition \ref{prop:1}. Similarly, the scenario optimization problem defined in Definition \ref{def:removefromset} produces the reachable set:
    \begin{equation}
        \hat{\mathcal{R}}^*_{SO}(\theta) = \{ \delta \in \mathbb{R}^{n_x}: g(\delta, \theta) \leq \overline{c^*}\} \notag.
    \end{equation}
    Since $\overline{c^*}$ is also computed as the largest $S-k$ value, $\overline{c^*} = \hat{q}$.
    Therefore, $\hat{\mathcal{R}}_{CP}(\theta) = \hat{\mathcal{R}}_{SO}(\theta)$.
    \qed
\end{proof}
We construct conformal and scenario programs with identically parameterized PAC guarantees in the form of Eq.~\eqref{eq:genpacbound} with equivalent reachable sets, $\hat{\mathcal{R}}^*_{SO} = \hat{\mathcal{R}}^*_{CP}$. 
Explicitly, given Eq. \eqref{eq:conformalbound} and Theorem \ref{thm:samplesdiscarded}, the following guarantees hold:
    \begin{equation}
        P^K(V(\hat{\mathcal{R}}^*_{CP}(\theta)| \{\delta\}^K_{i=1}) \leq \alpha) \geq 1 - \beta, 
    \end{equation} 
    \begin{equation}
        P^N(V(\hat{\mathcal{R}}^*_{SO}(\theta)) \leq \epsilon) \geq 1 - \beta.
    \end{equation}

\textbf{Remark 2.} While conformal prediction and scenario optimization render identically parameterized PAC bounds, the guarantees have different interpretations.
Conformal prediction bounds the probability that a new observation falls outside of the predictor produced by the algorithm, conditioned on the calibration set---an ex ante guarantee.
Scenario optimization bounds the true error of the learned object given the observed dataset---an ex post guarantee. 

\section{Empirical Results}\label{sec:empiricaleval}
In our empirical evaluation, we focus on forward reachable set estimates. 
We consider the nonlinear Duffing oscillator, which exhibits chaotic behavior as described in \cite[Sec. 4.1]{pmlr-v242-dietrich24a} and approximate reachable sets through scenario optimization with ellipsoids \cite[Eq. 8]{pmlr-v120-devonport20a}. 
We generate reachable set estimates with $N=1500$ samples and $\beta=10^{-9}$, yielding conservative and accurate initial parameterization $\theta$. 
Conformal prediction and sample-discarding scenario optimization adjust the threshold of the reachable set given a fixed $\theta$.

\subsection{Conformal Prediction and Holdout Method}

\begin{table}[b]
\caption{Comparison of Conformal Prediction and Holdout
\vspace*{-4mm}}
\begin{center}
\footnotesize
\resizebox{\columnwidth}{!}{
\begin{tabular}{lccc}
\toprule
& \multicolumn{2}{c}{\textbf{Average Accuracy}} & \textbf{Reachable Set Volume} \\
\midrule
\textit{Holdout} & \multicolumn{2}{c}{$\epsilon$} & \\
\midrule
$\mathbf{M=1500}$ &\multicolumn{2}{c}{0.0226} & 19.5 \\
$\mathbf{M=15000}$   & \multicolumn{2}{c}{0.0073} & 19.5 \\
\midrule
\textit{Conformal Prediction} & \multicolumn{2}{c}{$\alpha$}  \\
\midrule
$\mathbf{K=1500}$ & \multicolumn{2}{c}{0.0280} & 19.1 - 19.7 \\
$\mathbf{K=15000}$ & \multicolumn{2}{c}{0.0084} & 19.4 - 19.6\\
\bottomrule
\vspace{-7mm}
\end{tabular}}
\label{tab:accuracy}
\end{center}
\end{table}

We compare the holdout method and conformal prediction over $50$ trials for different test set sizes ($K=1500$ and $K=15000$), reporting the average accuracy and reachable set volumes in Table \ref{tab:accuracy}. 
We evaluate the holdout method and compute a binomial tail inversion with the observed $k_{holdout}$.
For conformal prediction, we take $k_{CP} = k_{holdout}$, use $k_{CP}$ to adjust the reachable set, and compute the corresponding inverse beta CDF, yielding the reachable set accuracy. The resulting conformal bound utilizes $k_{holdout}$ by weakening the accuracy of $\alpha$ \eqref{eq:betaaccuracies} by a factor of $\frac{\beta}{K+1}$.

\subsection{Split Conformal Prediction with Training-Conditional Coverage and Sample-Discarding Scenario Optimization}
\label{sec:empiricalcpso}

\addtolength{\textheight}{-10mm}   

Jointly solving Eq. \eqref{eq:calibcard} and Eq. \eqref{eq:alpha} is numerically unstable due to the magnitude of the optimization variables.
Instead, we demonstrate that given large quantities of i.i.d. data, the two methods yield almost identical reachable sets. 

First, we define a conformal program with $\beta=10^{-9}$, $\alpha=0.05$, and calibration set size $K$ is calculated using code provided by \cite{angelopoulos2022gentleintroductionconformalprediction}. Since training-conditional coverage for any fixed calibration set is not exactly $1- \alpha$, we instead compute the number of calibration samples required to achieve coverage close to this value. Following \cite{angelopoulos2022gentleintroductionconformalprediction}, we target $1-\alpha \pm 0.005$, resulting in $K=72,347$. For the corresponding scenario program (Definition~\ref{def:removefromset}), we take $N = K$, and $\epsilon=\alpha$, which results in $k=3,231$ discarded samples. Both methods produce nearly identical reachable sets when a large amount of i.i.d. data is available, reducing the original ellipsoid volume by $2.6$.

We now relax the parameterization to reflect more typical conformal usage. Specifically, we take $\beta=10^{-9}$, $\alpha=0.05$, and compute $K = 1,047$ to achieve $1-\alpha \pm 0.05$ coverage. The corresponding scenario program is defined with $\beta=10^{-9}$, $\epsilon=\alpha$, $N=K$, such that $k = 6$ scenarios are discarded. Fig. \ref{fig:relaxeduse} presents the resulting reachable sets: the conformal set adapts to account for coverage uncertainty, while the scenario set remains nearly unchanged. 

\begin{figure}[tb]
    \centering
    \resizebox{\columnwidth}{!}{
        \begin{tikzpicture}
\definecolor{CaliforniaGold}{rgb}{0.992, 0.71, 0.082}
\definecolor{BerkeleyBlue}{rgb}{0.0, 0.149, 0.463}
\definecolor{darkblue}{rgb}{0.0, 0.13, 0.28}
\definecolor{rose}{rgb}{1.0, 0.0, 0.5}
\begin{axis}[
    width=\linewidth,
    height=0.75\linewidth,
    xlabel={\Large$x$},
    ylabel={\Large$y$},
    ytick={-3, -2, -1, 0, 1},
    xtick={-3, -2, -1, 0, 1, 2, 3},
    xmin=-3.25,xmax=3.25,
    ymin=-3,ymax=1.75,
    ytick pos=left,
    xtick pos=bottom,
    legend style={
        font=\normalsize,
        legend image post style={scale=0.3},
        at={(1.04,0.5)},
        anchor=west,
        legend columns=1,
        /tikz/every even column/.append style={column sep=0.5cm},
    },
    tick label style={font=\normalsize},
]
\def\cx{0.12733294279295096}
\def\cy{-0.558924713479874}
\def\mxx{3.0478703491507675}   
\def\mxy{-0.07452859853386984}   
\def\myx{-0.07452859853386984}   
\def\myy{2.049718972545293}   
\addplot[
    domain=0:360,
    samples=200,
    ultra thick,
    opacity=0.25,
    draw=BerkeleyBlue,
]
(
    {\cx + \mxx*cos(x) + \mxy*sin(x)},
    {\cy + \myx*cos(x) + \myy*sin(x)}
);
\def\cx{0.12733294279295096 }
\def\cy{-0.558924713479874}
\def\mxx{2.7880549242457566}
\def\mxy{-0.06817541507216265}
\def\myx{-0.06817541507216265}
\def\myy{1.8749908690562118}
\addplot[
    domain=0:360,
    samples=200,
    thick,
    dashed,
    draw=BerkeleyBlue,
    forget plot,
]
(
    {\cx + \mxx*cos(x) + \mxy*sin(x)},
    {\cy + \myx*cos(x) + \myy*sin(x)}
);
\addlegendentry{Conformal Prediction}
\def\cx{0.01975192526006611}
\def\cy{-0.5351776361833893}
\def\mxx{3.0708122697704128}
\def\mxy{-0.12717321694539882}
\def\myx{-0.1271732169453989}
\def\myy{2.0075616988321814}
\addplot[
    domain=0:360,
    samples=200,
    ultra thick,
    opacity=0.25,
    draw=rose,
]
(
    {\cx + \mxx*cos(x) + \mxy*sin(x)},
    {\cy + \myx*cos(x) + \myy*sin(x)}
);
\addlegendentry{Scenario Optimization}
\def\cx{0.01975192526006611}
\def\cy{-0.5351776361833893}
\def\mxx{3.056070103190525}
\def\mxy{-0.12656269159119035}
\def\myx{-0.12656269159119043}
\def\myy{1.9979239201653014}
\addplot[
    domain=0:360,
    samples=200,
    thick,
    dashed,
    draw=rose,
]
(
    {\cx + \mxx*cos(x) + \mxy*sin(x)},
    {\cy + \myx*cos(x) + \myy*sin(x)}
);
\end{axis}
\end{tikzpicture}
        }
    \caption{Parameters: $\beta=10^{-9}; \alpha=\epsilon=0.05; K=N=1,047; k=6$. 
    Conformal prediction (blue) and the scenario (pink) program (Definition~\ref{def:removefromset}) adjust the ellipsoidal level set (dashed). Due to the high accuracy of the original reachable set (solid), conformal reduces the set volume by 3.2 to achieve marginal coverage and account for coverage uncertainty. In contrast, the scenario program only decreases volume by 0.2.
    \vspace{-5mm}}
    \label{fig:relaxeduse}
\end{figure}

\subsection{Sample and Computational Complexity}
To avoid excessive uncertainty during calibration,
conformal prediction requires more i.i.d. data than the holdout method or scenario optimization to achieve PAC guarantees.
As a result, either marginal coverage suffers or the resulting reachable set shrinks substantially, as seen in Fig.~\ref{fig:relaxeduse}. 
Additionally, the main benefit of conformal prediction is efficient computation of the quantile and adaption of the reachable set. However, the one-dimensional scenario program (Definition~\ref{def:removefromset}) mirrors this runtime speed.
As a result, the holdout method and scenario optimization are likely favorable for data-driven reachability analysis.

\section{Conclusion}
We established a formal connection between the PAC guarantees of conformal prediction, scenario optimization, and the holdout method.
Additionally, we demonstrated the practical implications of these relationships through an empirical case study on data-driven reachable sets. 
For conformal prediction, achieving PAC bounds entails a significant tradeoff between coverage guarantees, sample complexity, and reachable set volume. In contrast, the holdout method and scenario optimization directly provide bounds on the true error rate of the learned predictor and are therefore less susceptible to this tradeoff. 
Given a poor reachable set prediction, conformal prediction can provide a useful refinement by adapting the prediction set to achieve marginal coverage. However, when PAC bounds are desired, one-dimensional scenario optimization with sample discarding provides an equivalent adjustment while requiring fewer samples.
Otherwise, it is generally beneficial to (i) use the holdout method to obtain tight probabilistic bounds with minimal assumptions or (ii) use scenario optimization to formulate a well-posed optimization problem that jointly determines the predictor and its associated guarantee.



\bibliography{ref} 

@InProceedings{l4dc-comparison,
  title = 	 {Verification of neural reachable tubes via scenario optimization and conformal prediction},
  author =       {Lin, Albert and Bansal, Somil},
  booktitle = 	 {Learning for Dynamics \& Control Conference (L4DC)},
  pages = 	 {719--731},
  year = 	 {2024}
}

@article{scenario-sample-discard,
	author = {Campi, M. C. and Garatti, S.},
	journal = {Journal of Optimization Theory and Applications},
	number = {2},
	pages = {257--280},
	title = {A Sampling-and-Discarding Approach to Chance-Constrained Optimization: {Feasibility} and Optimality},
	volume = {148},
	year = {2011}
}

@INPROCEEDINGS{dietrich2025datadrivenreachabilityscenariooptimization,
  author={Elizabeth Dietrich and Devonport, Rosalyn and Tu, Stephen and Arcak, Murat},
  booktitle={Conf. on Decision and Control (CDC)}, 
  title={Data-Driven Reachability with Scenario Optimization and the Holdout Method}, 
  year={2025},
  pages={3925--3931}
  }

@article{angelopoulos2022gentleintroductionconformalprediction,
      title={A Gentle Introduction to Conformal Prediction and Distribution-Free Uncertainty Quantification}, 
      author={Anastasios N. Angelopoulos and Stephen Bates},
      year={2022},
      journal={arXiv:2107.07511}
}

@article{Jowett1963,
  author  = {G. H. Jowett},
  title   = {The Relationship Between the Binomial and {F} Distributions},
  journal = {Journal of the Royal Statistical Society. Series D (The Statistician)},
  volume  = {13},
  number  = {1},
  pages   = {55--57},
  year    = {1963},
  doi     = {10.2307/2986663}
}

@InProceedings{pmlr-v25-vovk12,
  title = 	 {Conditional Validity of Inductive Conformal Predictors},
  author = 	 {Vovk, Vladimir},
  booktitle = 	 {Asian Conference on Machine Learning},
  pages = 	 {475--490},
  year = 	 {2012},
}

@article{dembo1991scenario,
  title={{Scenario optimization}},
  author={Dembo, Ron S},
  journal={Annals of Operations Research},
  volume={30},
  pages={63--80},
  year={1991},
  publisher={Springer}
}

@article{angelopoulos2025theoreticalfoundationsconformalprediction,
      title={Theoretical Foundations of Conformal Prediction}, 
      author={Anastasios N. Angelopoulos and Rina Foygel Barber and Stephen Bates},
      year={2025},
      journal={arxiv:2411.11824}
}

@article{lindemann2025formalverificationcontrolconformal,
  author={Lindemann, Lars and Zhao, Yiqi and Yu, Xinyi and Pappas, George J. and Deshmukh, Jyotirmoy V.},
  journal={IEEE Control Systems}, 
  title={Formal Verification and Control With Conformal Prediction: Practical Safety Guarantees For Autonomous Systems}, 
  year={2025},
  volume={45},
  number={6},
  pages={72--122},
  doi={10.1109/MCS.2025.3611545}
}

@book{tempobook,
author = {Tempo, Roberto and Calafiore, Giuseppe and Dabbene, Fabrizio},
title = {Randomized Algorithms for Analysis and Control of Uncertain Systems: With Applications},
year = {2012},
isbn = {1447146093},
publisher = {Springer},
edition = {2nd}
}

@book{hardtrecht2022patterns,
  author = {Moritz Hardt and Benjamin Recht},
  title = {Patterns, predictions, and actions: Foundations of machine learning},
  year = {2022},
  publisher = {Princeton University Press}
}

@article{langford2005tutorial,
  title={Tutorial on Practical Prediction Theory for Classification},
  author={Langford, John},
  journal={Journal of Machine Learning Research},
  volume={6},
  pages={273--306},
  year={2005}
}

@INPROCEEDINGS{osullivan2025bridgingconformalpredictionscenario,
  author={O'Sullivan, Niall and Romao, Licio and Margellos, Kostas},
  booktitle={Conf. on Decision and Control (CDC)}, 
  title={Bridging conformal prediction and scenario optimization}, 
  year={2025},
  pages={6114--6121}}

@article{pacbound,
author = {Valiant, L. G.},
title = {{A theory of the learnable}},
year = {1984},
publisher = {Association for Computing Machinery},
address = {New York, NY, USA},
volume = {27},
number = {11},
issn = {0001-0782},
journal = {Commun. ACM},
pages = {1134–-1142},
}

@InProceedings{Griffioen2023,
  author    = {Griffioen, Paul and Arcak, Murat},
  booktitle = {Conf. on Decision and Control (CDC)},
  title     = {Data-Driven Reachability Analysis for {Gaussian} Process State Space Models},
  year      = {2023},
  pages     = {4100--4105},
  doi       = {10.1109/CDC49753.2023.10383270},
}

@InProceedings{Hashemi2023,
  author    = {Hashemi, Navid and Qin, Xin and Lindemann, Lars and Deshmukh, Jyotirmoy V.},
  booktitle = {Conf. on Decision and Control (CDC)},
  title     = {Data-Driven Reachability Analysis of Stochastic Dynamical Systems with Conformal Inference},
  year      = {2023},
  pages     = {3102--3109},
  doi       = {10.1109/CDC49753.2023.10384213},
}

@InProceedings{Tebjou2023,
  author    = {Tebjou, Abdelmouaiz and Frehse, Goran and Chamroukhi, Fa"{i}cel},
  booktitle = {Symposium on Conformal and Probabilistic Prediction with Applications},
  title     = {Data-driven Reachability using {Christoffel} Functions and Conformal Prediction},
  year      = {2023},
  pages     = {194--213},
}

@InProceedings{pmlr-v242-dietrich24a,
  title = 	 {Nonconvex Scenario Optimization for Data-Driven Reachability},
  author =       {Dietrich, Elizabeth and Devonport, Alex and Arcak, Murat},
  booktitle = 	 {Learning for Dynamics \& Control Conference (L4DC)},
  pages = 	 {514--527},
  year = 	 {2024},
}

@Article{Shafer2008,
  author     = {Shafer, Glenn and Vovk, Vladimir},
  journal    = {J. Mach. Learn. Res.},
  title      = {{A Tutorial on Conformal Prediction}},
  year       = {2008},
  issn       = {1532-4435},
  pages      = {371–-421},
  volume     = {9},
  numpages   = {51},
  publisher  = {JMLR.org},
}

@article{nonconvex24,
	author = {Garatti, Simone and Campi, Marco C.},
	date = {2024/04/08},
	doi = {10.1007/s10107-024-02074-3},
	id = {Garatti2024},
	isbn = {1436-4646},
	journal = {Mathematical Programming},
	title = {{Non-convex scenario optimization}},
	year = {2024}}

@InProceedings{pmlr-v120-devonport20a,
  title = 	 {Estimating Reachable Sets with Scenario Optimization},
  author =       {Devonport, Alex and Arcak, Murat},
  booktitle = 	 {Proceedings of the 2nd Conference on Learning for Dynamics and Control},
  pages = 	 {75--84},
  year = 	 {2020},
  volume = 	 {120},
  month = 	 {10--11 Jun},
  publisher =    {PMLR},

}

@InProceedings{Akametalu2014,
  author    = {Akametalu, Anayo K. and Fisac, Jaime F. and Gillula, Jeremy H. and Kaynama, Shahab and Zeilinger, Melanie N. and Tomlin, Claire J.},
  booktitle = {Conf. on Decision and Control (CDC)},
  title     = {Reachability-based safe learning with {Gaussian} processes},
  year      = {2014},
  pages     = {1424--1431},
  doi       = {10.1109/CDC.2014.7039601},
}

@ARTICLE{alanwar23-TAC,
  author={Alanwar, Amr and Koch, Anne and Allgöwer, Frank and Johansson, Karl Henrik},
  journal={IEEE Transactions on Automatic Control}, 
  title={Data-Driven Reachability Analysis From Noisy Data}, 
  year={2023},
  volume={68},
  number={5},
  pages={3054--3069},
  doi={10.1109/TAC.2023.3257167}}

@Article{Park2024,
  author   = {Park, Hyunsang and Vijay, Vishnu and Hwang, Inseok},
  journal  = {IEEE Control Systems Letters},
  title    = {Data-Driven Reachability Analysis for Nonlinear Systems},
  year     = {2024},
  pages    = {2661--2666},
  volume   = {8},
  doi      = {10.1109/LCSYS.2024.3510595},
}

@INPROCEEDINGS{campicdc25,
  author={Carè, Algo and Campi, Marco C. and Garatti, Simone},
  booktitle={Conf. on Decision and Control (CDC)}, 
  title={Post-Design Verification in the Scenario Approach}, 
  pages={5782--5787},
  year={2025}}

@INPROCEEDINGS{garatticdc23,
  author={Garatti, Simone and Campi, Marco C.},
  booktitle={Conf. on Decision and Control (CDC)}, 
  title={Scenario Optimization with Constraint Relaxation in a Non-Convex Setup: {A} Flexible and General Framework for Data-Driven Design}, 
  year={2023}}

@INPROCEEDINGS{Coppola2024,
  author={Coppola, Rudi and Peruffo, Andrea and Lindemann, Lars and Mazo, Manuel},
  booktitle={European Control Conference (ECC)}, 
  title={Scenario Approach and Conformal Prediction for Verification of Unknown Systems via Data-Driven Abstractions}, 
  year={2024},
  pages={558--563},
  doi={10.23919/ECC64448.2024.10590827}}

@article{coppola2025trainingconditionalconformalpredictionbinomial,
title={On Training-Conditional Conformal Prediction and Binomial Proportion Confidence Intervals},
author={Rudi Coppola and Manuel Mazo Espinosa},
journal={Transactions on Machine Learning Research},
year={2025},
}
\bibliographystyle{IEEEtran}

\end{document}